\documentclass[lettersize,journal]{IEEEtran}

\usepackage{cite}
\usepackage{algorithmic}
\usepackage{amsmath,amsfonts}
\usepackage{algorithm}
\usepackage{array}
\usepackage[caption=false,font=normalsize,labelfont=sf,textfont=sf]{subfig}
\usepackage{textcomp}
\usepackage{stfloats}
\usepackage{url}
\usepackage{verbatim}
\usepackage{graphicx}
\usepackage{amsthm}
\usepackage{booktabs}
\usepackage{xcolor}
\newtheorem{definition}{Definition}
\begin{document}

\setlength{\abovecaptionskip}{1pt}
\setlength{\belowcaptionskip}{0pt}

\title{Control Analysis and Design for Autonomous Vehicles Subject to Imperfect AI-Based Perception}

% \author{Tao Yan
\author{Tao Yan, Zheyu Zhang, Jingjing Jiang, and Wen-Hua Chen
% \author{Tao Yan, Zhe Xu, S. Andrew Gadsden,~\IEEEmembership{Senior Member,~IEEE,} Simon X. Yang,~\IEEEmembership{Senior Member,~IEEE}
        % <-this % stops a space
\thanks{Tao Yan, Zheyu Zhang, Jingjing Jiang, and Wen-Hua Chen are with the Department of Aeronautical and Automotive Engineering, Loughborough University, Loughborough, LE113TU, U.K.\{t.yan, z.zhang8, j.jiang2,
w.chen\}@lboro.ac.uk }}

% \thanks[footnoteinfo]{This paper was not presented at any IFAC 
% meeting. Corresponding author M.~T.~Cicero. Tel. +XXXIX-VI-mmmxxi. 
% Fax +XXXIX-VI-mmmxxv.}

% \author[Paestum]{Marcus Tullius Cicero}\ead{cicero@senate.ir},    % Add the 
% \author[Rome]{Julius Caesar}\ead{julius@caesar.ir},               % e-mail address 
% \author[Baiae]{Publius Maro Vergilius}\ead{vergilius@culture.ir}  % (ead) as shown

% \address[Paestum]{Buckingham Palace, Paestum}  % Please supply                                              
% \address[Rome]{Senate House, Rome}             % full addresses
% \address[Baiae]{The White House, Baiae}        % here.

\maketitle

\begin{abstract}
Safety is a critical concern in autonomous vehicle (AV) systems, especially when AI-based sensing and perception modules are involved. However, due to the black box nature of AI algorithms, it makes closed-loop analysis and synthesis particularly challenging, for example, establishing closed-loop stability and ensuring performance, while they are fundamental to AV safety. To approach this difficulty, this paper aims to develop new modeling, analysis, and synthesis tools for AI-based AVs. Inspired by recent developments in perception error models (PEMs), the focus is shifted from directly modeling AI-based perception processes to characterizing the perception errors they produce. Two key classes of AI-induced perception errors are considered: misdetection and measurement noise. These error patterns are modeled using continuous-time Markov chains and Wiener processes, respectively. By means of that, a PEM-augmented driving model is proposed, with which we are able to establish the closed-loop stability for a class of AI-driven AV systems via stochastic calculus. Furthermore, a performance-guaranteed output feedback control synthesis method is presented, which ensures both stability and satisfactory performance. The method is formulated as a convex optimization problem, allowing for efficient numerical solutions. The results are then applied to an adaptive cruise control (ACC) scenario, demonstrating their effectiveness and robustness despite the corrupted and misleading perception.
\end{abstract}


\section{Introduction}
\IEEEPARstart{A}{UTONOMOUS} vehicles (AVs), due to their automated and intelligent nature, find a broad range of applications not only in industry areas but also for public uses \cite{Nascimento,ELALLID20227366,8818311,Yan215}. One supportive and fundamental technology for AVs to operate at a higher autonomy is attributed to the wide deployment of artificial intelligence (AI) based perception algorithms and models. These AI systems are capable of efficiently processing a vast amount of raw sensory data, such as images and point clouds captured by on-board sensors, and converting them into low-dimensional perceptual information (e.g., positions and velocities of surrounding objects). This processed information is then fed into control modules to facilitate intelligent decision-making and planning \cite{MaYifang, ZHANG2, PREETI2024124664, Everingham}. On the other hand, the behavior and reliability of these components could pose new challenges for AVs to safely and satisfactorily achieve their goals, for example, following a car ahead and adaptive cruise at a constant speed without any collisions.

{In automotive applications, almost all of the tasks are inherently safety-critical. As a result, significant research efforts have been devoted to improving and ensuring the safety of various AI-driven automated driving systems (ADSs)\cite{Zhang45,Nascimento,8818311,ZHANG2,Piazzoni,PREETI2024124664}. In industry, one of the most common and cost-effective methods for evaluating the behavior of AVs is through virtual simulation testing, by which many typical scenarios and contexts that AVs may encounter in reality can be easily set and tested \cite{ZHANG2,Rong9,Piazzoni,FengShuo}. However, this method can be overly idealized and may lead to the results that deviate significantly from the real-life situations. For example, the simulators and sensor modeling are often simplified in order to ensure computation efficiency \cite{Rong9,9151398,9176787}.  To enhance the fidelity of simulations while maintaining low computational costs, a promising alternative approach, called perception error model (PEM), has  recently been proposed \cite{ijcai}. The PEM incorporates sensing and perception pipeline into a single module, therefore circumventing the need for high-precision sensor modeling. Interestingly, the goal of PEMs is not to seek a better modeling for the input-output mapping, but rather to explicitly model the underlying AI-induced uncertainties. It is shown in \cite{Piazzoni,mitra2018towards} that PEMs can serve as efficient and effective  alternatives to replacing sensing and perception systems. These models are then utilized to study the impact of perception errors on safety of AVs \cite{Piazzoni,berkhahn2021traffic}. Similarly, PEMs are used to estimate the failure probabilities in an automated emergency braking system \cite{Innes}, and the authors also argue that the photorealistic, high-fidelity simulations may not be essential for evaluating the ADS safety, as long as the key characteristics and the probabilistic distributions of perception errors are sufficiently represented.}

{The aforementioned works mainly focus on generating simulated driving scenarios that better replicate the real-world situations, in which the properties of AI-based autonomous vehicles can be tested. With the aid of digital testbeds, various automated driving polices are able to be explored. To achieve desired driving tasks, methods such as rule-based approaches and intelligent driver models (IDMs) have been widely utilized in industry, owing to their effectiveness and simplicity \cite{Albeaik,CHEN,xiao2021rule,aksjonov2021rule}. An iterative linear quadratic regulator is designed to solve nonlinear motion planning with constraints\cite{Chen2019}. More recently, learning-based approaches have gained popularity for their adaptability to unknown and unseen situations \cite{liu2021decision}. In these studies, while the uncertainty on environments can be considered, they usually assume ideal perception. However, as highlighted earlier, modern AI-based perception systems frequently suffer from errors and even certain sort of faulty outputs (e.g., misdetection, misclassification, occlusion, etc.). Ignoring these issues in design can lead to unsafe and risky driving behavior in practice \cite{pandharipande,Liu,ZhangZheyu}. Therefore, designing driving policies that are aware of erroneous perception is of significant importance. In \cite{artunedo2020motion},  the uncertainty in localization is considered in the stage of motion planning. An estimation strategy is proposed to compensate for perceptual errors and leads to efficient regulation for connected and automated traffic \cite{li2021cooperative}. In \cite{bonzanini2021}, perception-aware model predictive control is studied to improve robustness to noisy sensor inputs. Besides Gaussian noise, the effect of intermittent packet dropout over communication networks is addressed in the field of vehicular platooning\cite{zhao2020stability,acciani2018cooperative,gordon2021comparison}. Kalman filtering based techniques are widely utilized in such cases against the intermittent observations \cite{liu2016stochastic,villenas2023kalman}. 
}

While extensive efforts has been made to validate the safety and performance of AI-based autonomous vehicles, the results are largely based on synthetic digital simulations \cite{ijcai,mitra2018towards,Piazzoni,berkhahn2021traffic,Innes}. Due to the numerical nature of those approaches, it is impossible to cover all potential driving scenarios and, thus, makes the evaluation inherently restrictive. This motivates the need for a formal framework that enables to study and reveal underlying interactions analytically between performance of AI-driven perception systems and autonomous driving. In terms of driving policy design, earlier works can manage either noisy measurements or intermittent communications independently \cite{artunedo2020motion,bonzanini2021,zhao2020stability,acciani2018cooperative}, yet limited attention has been given to the coupled effects of both factors and, indeed, there remains a gap to address AI-based perception in autonomous vehicles. In addition, most of these schemes primarily concentrate on algorithmic aspects and lack rigorous discussions on resultant vehicle's closed-loop properties.

To bridge this gap, this paper aims to provide a formal framework, by which the closed-loop properties of AI-based autonomous vehicles can be analytically established. The proposed framework enables also control synthesis that explicitly accounts for those effects caused by the AI-driven perception systems and ensures desired behavior across all scenarios. It is worth noting that the closed-loop stability and performance are fundamental to the safety of autonomous vehicles. Stability guarantees that the vehicle remains controllable and behaves robustly to various errors and disturbances. Performance (e.g., response speed and steady-state errors) means that the AV can accurately follow desired trajectories, maintain safe distances, and avoid collisions. Without a stable and well-performing closed-loop system, the safety of AI-based autonomous vehicles cannot be assured. The contributions are threefold:
{\begin{itemize}
    \item To address the challenge of representing AI-based perception systems which normally are formed by deep neural networks\cite{PREETI2024124664}, a novel modeling means is introduced based on perception error models (PEMs); rather than focusing on input-output mappings, this approach emphasizes the description of AI-induced perception errors \cite{ijcai,Piazzoni,Innes,berkhahn2021traffic}. Specifically, two typical AI-induced error patterns, that is, misdetection and measurement noise, are particularly concentrated and modeled by Markov chains and Gaussian processes, respectively, as these two have been demonstrated to have significant impacts on behavior of AI-based vehicles\cite{ijcai,Piazzoni,mitra2018towards}.  With this modeling, a PEM-augmented driving model (PEM-ADM) is proposed to describe the error behavior of certain autonomous driving systems. It is important to note that unlike traditional state-space models, PEM-ADM incorporates a heterogeneous source of uncertainty in its measurement equations.
    \item Based on the PEM-ADM, the effects caused by the AI-based perception are formally investigated. More precisely, the closed-loop properties of AI-based autonomous vehicles, such as stability and steady-state accuracy, are analyzed. A sufficient condition is developed to establish the stochastic stability of the closed-loop in the mean square sense. Moreover, an upper bound on the steady-state error is provided,  
    and this offers insights into  how imperfect perception and controls could affect overall driving performance.
    \item Two efficient output feedback control synthesis methods are proposed. One allows to design controllers that ensure stochastic stability, even though the measurement information is concurrently affected by Markovian misdetection and Gaussian noise. Furthermore, the second control synthesis approach not only guarantees closed-loop stability, but also allows for specification of desired driving performance, for example, faster convergence and improved steady-state accuracy. The   presented performance assured method is formulated as a constrained convex optimization problem with linear matrix inequalities. Therefore, it can be efficiently solved off-line by existing numerical solvers. The developed results are validated through an adaptive cruise control case study, demonstrating their effectiveness and robustness under corrupted  and misleading perception.
\end{itemize}}

The rest of the paper is organized as follows: Section~\ref{sec2} presents the modeling for AI-based AV systems and states the problem. The main results on control analysis and synthesis are developed in Section~\ref{sec3}. Results are then applied to an adaptive cruise control problem in Section~\ref{sec4}. Section~\ref{sec5} provides illustrative examples, and Section~\ref{sec6}  concludes the paper.

\textit{Notation:} Denote $d x$ as the differential of $x$, and $\dot{x}(t)$ the time derivative of $x$. The symbol $\mathbb{E}[\cdot]$ stands for the mathematical expectation of a random variable, and $\mathcal{L}$ stands for the infinitesimal operator of a stochastic process. Denote the transpose of a matrix $A$ by $A^T$. A symmetric matrix $A > 0$ means $A$ is a positive definite matrix and, conversely, $A < 0$ a negative definite matrix. $\lambda_{max}(A)$, $\lambda_{min}(A)$ denote the maximum and minimum eigenvalues of $A$, respectively. $A^{-1}$ denotes the inverse of $A$. The pseudo inverse is denoted by $A^+$. Function $\text{diag}(\cdot)$ defines a diagonal matrix. The trace of $A$ is denoted by $\text{tr}(A)$.

\begin{figure}[!tbp]
\centering
\includegraphics[width=0.45\textwidth,trim=0.85cm 0.3cm 0.9cm 0.3cm, clip]{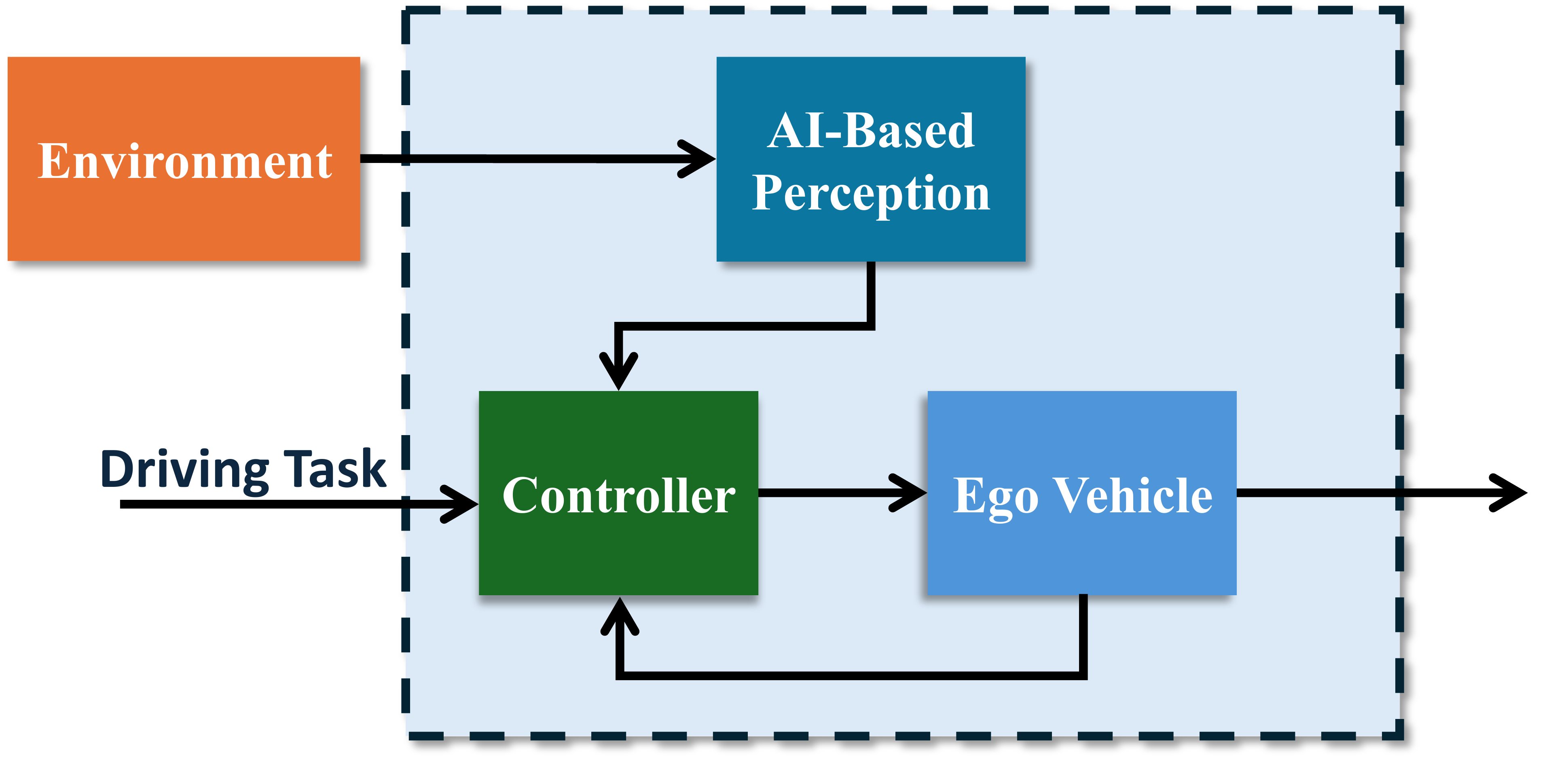}
\caption{Schematic of signal flows within an ADS.}
\label{fig:01}
\end{figure}
\section{modeling and problem statement} \label{sec2}
\subsection{Perception Error Model}
In practice, the AI-based perception process typically comprises two components, that is, sensing and perception (S\&P) systems. The outputs of the S\&P are subsequently fed into the control modules to make informed decisions, as illustrated in Fig. \ref{fig:01}. However,  due to the complicated mechanism of sensing and AI perception algorithms, it remains a significant challenge to incorporate S\&P systems into closed-loop analysis for autonomous driving tasks.

To overcome this difficulty, instead of looking for the modeling for S\&P input-output, PEMs focus on characterizing the statistical properties of AI-induced perception errors \cite{ijcai,Piazzoni,Innes,berkhahn2021traffic}. More specifically, the outputs of the AI-based perception systems are represented by 
\begin{align}
    PEM(\mathcal{W}) = \mathcal{W}+ \mathcal{\varepsilon} \label{eq:1}
\end{align}
where $ \mathcal{W}$ denotes the actual world or the environment where the ego vehicle operates in, and $\varepsilon$ captures the underlying errors and uncertainties caused by the AI-based perception process. In this formulation, the sole need is to represent error characteristics~$\varepsilon$ of perception instead of the whole S\&P input-output.

{This paper will particularly focus on two fundamental classes of AI-induced perception errors, namely, misdetection and measurement noise. These two factors have been demonstrated to have major impacts on the safety and performance of current AI-embedded AV systems \cite{ijcai,Piazzoni,mitra2018towards}. In the following, we will present the modeling for AI-driven automated driving systems.}

\subsection{PEM-Augmented Driving Model}
{As illustrated in Fig.~\ref{fig:01}, many autonomous driving tasks require interactions with environments whose dynamics can be represented as
\begin{align}
    \dot{x}^{env} = f_1(x^{env},w^{env}) \label{eq:2:r1}
\end{align}
where $f_1(\cdot,\cdot)$ governs the evolution of the environment; $x^{env}$ denotes the state of the environment (e.g., location, velocities, or accelerations); $w^{env}$ represents the uncertainty applied on the environment.}

The motion of the ego vehicle is described by
\begin{align}
    \dot{x}^{ego} = f_2(x^{ego},u^{ego},w^{ego}) \label{eq:3:r1}
\end{align}
where $f_2(\cdot,\cdot,\cdot)$ denotes the dynamics; $x^{ego}$, $u^{ego}$ represent the state and control actions of the ego vehicle, respectively, and $w^{ego}$ signifies the enforced uncertainty.

The goal of ADSs is usually to drive the ego vehicle to maintain certain synchronization with the environments (e.g., car following, lane changing and keeping) under a well-structured driving policy. To facilitate analysis and synthesis for an AI-driven ADS, we propose the following PEM-augmented driving model (PEM-ADM):
\begin{align}
    \dot{x}(t) &= A(i) x(t) + B(i) u(t), \label{eq_1_} \\
    y(t) &= C(i) x(t) + D(i) \omega (t), \nonumber \\
    i &= r(t), ~ i \in \mathcal{I} = \{1, \cdots, N\} \label{eq_1}
\end{align}
where $A(i) \in \mathbb{R}^{n\times n}$, $B(i) \in \mathbb{R}^{n\times m}$, $C(i) \in \mathbb{R}^{p\times n}$, and $D(i) \in \mathbb{R}^{p\times m}$ are the task-related system matrices yet subject to random switching, as indexed by $i$, to account for the potential faults and misdetection; in particular, $i$ is orchestrated by $r(t)$ satisfying a continuous-time Markov chain  \cite{boukas2007stochastic}, defined by a generator matrix $Q= [q_{ij}] \in \mathbb{R}^{N\times N}$. $x(t)\in \mathbb{R}^{n}$ represents the error state of an ADS, which usually captures the deviation between some dimensions of $x^{env}(t)$ and $x^{ego}(t)$, and the driving input is denoted by $u(t) \in \mathbb{R}^{m}$ often equal to $u^{ego}(t)$. The measurements $y(t) \in \mathbb{R}^{p}$ used for control are the erratic version of $x(t)$ due to the effects of AI-based perception; $\omega(t)\in \mathbb{R}^{p}$ denotes the perception or measurement noise independent of $t$, each entry of which satisfies Gaussian distribution $\mathcal{N}(0,1)$ and mutually independent as well.
\newtheorem{remark}{Remark}

{The PEM-ADM is a control-oriented model for autonomous driving and consists of two main components: dynamic model~\eqref{eq_1_} describes the error behavior of the ADS, and AI-based measurement equation \eqref{eq_1} follows a PEM formulation~\eqref{eq:1}. It is important to note that the actual error evolution may not adhere to a linear representation. As indicated in \eqref{eq:2:r1} and \eqref{eq:3:r1}, the dynamics are usually nonlinear. However, numerous effective handling methods from robust and nonlinear control literature can be used to address this issue. For instance, estimators can be designed to approximate and counteract the nonlinear dynamics, rendering the remaining system linear. Alternatively, the proposed model can serve as a motion planner; that is, the outputs from our model could be interpreted as reference trajectories for lower-level controllers to track. The high-fidelity nonlinear dynamics can then be incorporated at this lower level.  Thus, assuming linear model~\eqref{eq_1_} is not overly restrictive. Our focus, on the other hand, is on handling the effects caused by the AI-based perception. As shown in \eqref{eq_1}, the resulting measurements are subject to Markov switching and Gaussian noise concurrently. The Markov chain determines which measurement mode is active at each sampling time among a given finite set; for example, in the simplest case, whether a measurement is available or not. Meanwhile, the Gaussian process injects noise into each measurement independently over time. These two features complicate the ADS analysis and synthesis and remain underexplored.}

% An important remark to anwer why liner modes are appraopriate.
% Note that it is not restrictive to consider the linear models, as this types models are much suited for the purpose of high-level motion planning; in other words, the results obtained in this respect could be served as the references for lower-level controllers to track.

% For simplicity, we here use a binary variable to capture the state of the perception systems, that is, $r(t) = 0, 1$, where $r=0$ indicates that the perception system loses tracking on objects, and $r=1$ means a successful tracking. Thus, the generator matrix $Q$ can be further given as
% \begin{align}
%     Q = \begin{bmatrix}
%     q_{00} & q_{01} \\
%     q_{10} & q_{11}
%     \end{bmatrix} \label{eq_2}
% \end{align}
% with $q_{01} $, $q_{10} > 0$ and  $q_{00} = -q_{01} $, $q_{11} = -q_{10} $.

The following concept of the stochastic stability will be dealt with throughout the paper.
% \begin{definition}
%     The trivial solution $x(t)  \equiv 0$ of a stochastic system is said to be asymptotically stable in mean square if 
%     \begin{align}
%         \mathbb{E} [x(t)^Tx(t)] \to 0 \quad \text{as} \quad t \to \infty \label{eq_3}
%     \end{align}
% \end{definition}

\begin{definition}
    The trivial solution $x(t)  \equiv 0$ of a stochastic system is said to be ultimately bounded stable (UBS) in mean square if there exists $M > 0$ such that
    \begin{align}
        \mathbb{E} [x(t)^Tx(t)] \le M \quad \text{as} \quad t \to \infty \label{eq_4}
    \end{align}
\end{definition}

\begin{figure}[!tbp]
\centering
\includegraphics[width=0.45\textwidth,trim=0cm 0.3cm 0.9cm 0.3cm, clip]{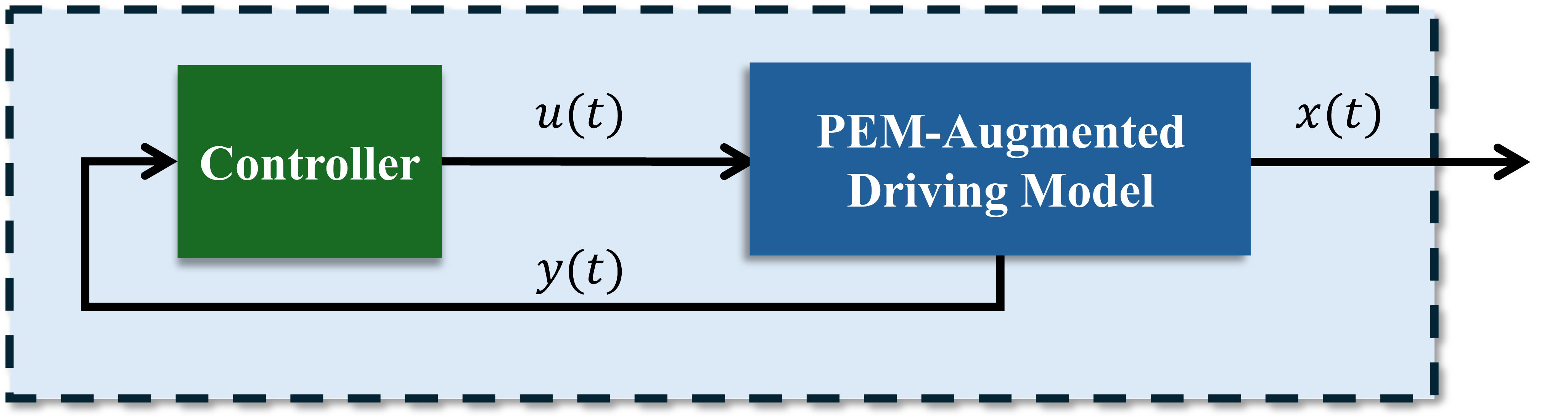}
\caption{Control block diagram for a PEM-ADM.}
\label{fig:02}
\end{figure}

 As illustrated in Fig.~\ref{fig:02},  the first goal of the article is intended to establish the closed-loop stochastic stability of the PEM-ADM with respect to a given output feedback controller that treats the AI-based measurement  $y(t)$ as the input. Furthermore, our second goal is devoted to the controller design, namely, offering a systematic procedure to find an output feedback controller for $u(t)$ that stochastically stabilizes the PEM-ADM system meanwhile with some guaranteed performance, such as fast convergence and high driving accuracy.

\section{Control Analysis and Synthesis} \label{sec3}
To achieve the above objectives, here it is focused on linear output feedback control as
\begin{align}
    u(t) = K(i) y(t) \label{eq_5}
\end{align}
where $K(i) \in \mathbb{R}^{m \times p}$ are given control gains for each mode $i$ ($i = 1,\cdots,N$).

\subsection{Closed-loop Stability Analysis}
This subsection will develop a sufficient condition under which the stochastically jump and uncertain system \eqref{eq_1} can be stabilized with the presented linear control structure \eqref{eq_5}. 

For this purpose, plugging in the output feedback control law \eqref{eq_5} into \eqref{eq_1} gives the resulting closed-loop: 
\begin{align}
    \dot{x}(t) = [ A(i)+ B(i)K(i) C(i) ]x(t) + B(i) K(i)  D(i)\omega(t) \label{eq:4}
\end{align}
It is convenient to rewrite \eqref{eq:4} into the stochastic differential form
\begin{align}
    d {x} (t) = [ A(i)+ B(i)K(i) C(i) ]x(t) dt + B(i) K(i)  D(i) dw(t) \label{eq_6}
\end{align}
where $dw(t) = \omega(t) dt$. It is simple to verify that $w(t)$ is the Wiener process, for which tools like Ito calculus are well-established.

Denote $A_{cl}(i) =  A(i) + B(i)K(i)C(i)$, $W(i) = B(i) K(i)  D(i)$. Then \eqref{eq_6} can be written as
\begin{align}
    d {x} (t) = A_{cl}(i) x(t) dt + W(i) dw(t)  \label{eq_7}
\end{align}

The following proposition summarizes the result of the closed-loop stability.
\newtheorem{proposition}{Proposition}
\begin{proposition} \label{thm1}
    The origin of the closed-loop system \eqref{eq_7} is UBS in mean square if there exist $P(i) >0$ such that the following matrix inequality holds
    \begin{align}
         A_{cl}(i)^T P(i) + P(i) A_{cl}(i) + \sum _{j=1}^{N} q _{ij} P(j) < 0, ~ \forall i \in \mathcal{I} \label{eq_8}
    \end{align}
\end{proposition}
\begin{proof}
    To help derive the result, a Lyapunov-like approach is adopted, and the Lyapunov function candidate is proposed as
    \begin{align}
        V(x(t),r(t)) = x(t) ^T P(i) x(t), ~ \text{when} ~ r(t) =i \label{eq_9}
    \end{align}
    where $P(i)$ is a symmetric positive definite matrix. Employing the infinitesimal operator on $V(x(t),r(t))$ together with \eqref{eq_7} yields
    \begin{align}
        \mathcal{L} V(x(t),r(t)) &= x(t)^TA_{cl}(i) ^ T V_x (x(t),i) + \sum _{j=1}^{N} q _{ij} V(x(t),j) \nonumber \\
        &~ + \frac{1}{2} \mathrm{tr}[W(i) ^T V_{xx} (x(t),i) W(i)] \nonumber \\
        &= x(t)^T [ A_{cl}(i)^T P(i) + P(i) A_{cl}(i) \nonumber \\
        &~ + \sum _{j=1}^{N} q _{ij} P(j)] x(t) + \mathrm{tr}[ W(i) ^T P(i) W(i) ] \label{eq_10}
    \end{align}
    The derivation for the first line of \eqref{eq_10} can be referred to  Appendix.  Let $Q(i) = A_{cl}(i)^T P(i) + P(i) A_{cl}(i) + \sum _{j=1}^{N} q _{ij} P(j)$, $\chi(i) = \mathrm{tr}[ W(i) ^T P(i) W(i) ]$ and denote $\lambda _{max}(Q(i)) $ as the minimum eigenvalue of $Q(i)$. Using the condition \eqref{eq_8}, for every $i \in \mathcal{I}$, \eqref{eq_10} can be bounded by
    \begin{align}
        \mathcal{L} V(x(t),r(t)) & \le \max_{i \in \mathcal{I}} \{\lambda _{max}(Q(i))\} x^T(t) x(t) \nonumber \\
        &~ + \max_{i \in \mathcal{I}} \{ \chi(i) \} \label{eq_11}
    \end{align}
    By Dynkin's formula and Fubini's Theorem, we see that
    \begin{align}
        \mathbb{E}[V(x(t),r(t))] - V(x_0,r_0) &= \mathbb{E}[\int_{t_0}^t \mathcal{L} V(x(s),r(s)) ds] \nonumber \\
        & = \int_{t_0}^t \mathbb{E}[\mathcal{L} V(x(s),r(s))] ds \label{eq_12}
    \end{align}
    Let  $\gamma _1 = - \max_{i \in \mathcal{I}} \{\lambda _{max}(Q(i))\}$; obviously, $\gamma _1 >0$ since $Q(i) < 0$ as per \eqref{eq_8}, and set $c_1 =  \max_{i \in \mathcal{I}} \{ \chi(i) \}$. Then differentiating \eqref{eq_12} and using \eqref{eq_11}, we get
    \begin{align}
        \frac{d}{dt} \mathbb{E}[V(x(t),r(t))] &=\mathbb{E}[\mathcal{L} V(x(t),r(t)) ] \nonumber \\
        & \le -\gamma _1 \mathbb{E}[x(t)^Tx(t)] +c_1 \nonumber \\
        & \le 0, ~\text{whenever} ~ \mathbb{E}[x(t)^T x(t)] \ge c_1/\gamma_1 \label{eq_13}
    \end{align} 
    On the other hand, by \eqref{eq_9} the following holds 
    \begin{align}
        \gamma_2 x(t) ^T x(t)  \le  V(x(t),r(t) ) \le \gamma_3 x(t) ^Tx(t) \label{eq_14}
    \end{align}
    where $\gamma_2 = \min _{i \in \mathcal{I}}\{\lambda_{min}(P(i))\}$ and $\gamma_3 = \max_{i \in \mathcal{I}} \{ \lambda_{max}(P(i)) \}$.
    
    Based on \eqref{eq_13}, together with \eqref{eq_14}, it is easy to see that for every $i \in \mathcal{I}$,
    \begin{align}
        \mathbb{E}[V(x(t),r(t))] &\le \frac{\gamma_3 c_1 }{\gamma_1} \label{eq_15}
    \end{align}
    Applying \eqref{eq_14} again yields
    \begin{align}
        \mathbb{E}[x(t)^T x(t)] & \le \frac{\gamma_3 c_1}{\gamma_1 \gamma_2}, \label{eq_16}
    \end{align}
    which shows that the origin of the system is UBS in mean square. The proof is complete.
\end{proof}

\begin{remark}
    In the above proof, except for the analysis of stochastic stability for a given controller, a lower bound on closed-loop steady state performance is derived as given in \eqref{eq_16}, which can be exploited in below for a performance guaranteed control design.
\end{remark}

\begin{remark} \label{rmk:2}
    By \eqref{eq_14} it is easy to see that $\mathbb{E}[x(t)^T x(t)]  \ge 1/\gamma _3 \mathbb{E}[V(x(t),r(t))]$. Therefore, the following is true in accordance with \eqref{eq_13}:
    \begin{align}
        \frac{d}{dt} \mathbb{E}[V(x(t),r(t))] \le - \frac{\gamma_1}{\gamma _3} \mathbb{E}[V(x(t),r(t))] + c_1
    \end{align}
    Hence, it can be estimated that the rate of the convergence is no smaller than the value of $\gamma_1/\gamma_3$ before system states enter the region as described by \eqref{eq_15}.
\end{remark}

\begin{remark} \label{rmk:3}
    An upper bound on $ \mathbb{E}[x(t)^T x(t) ] $ over the entire process can be obtained as well. This kind of bound is useful for safe critical tasks, in which the system states are required to restrict to some prescribed domain for all time. To see that, applying again \eqref{eq_13} and \eqref{eq_14}, it yields that whenever $\mathbb{E}[x(t)^T x(t)] \ge c_1/\gamma_1$, we have
    \begin{align}
         \gamma _2  \mathbb{E}[x(t)^T  x(t) ] \le \mathbb{E}[V(x(t),r(t))] \le V(x_0,r_0) \le \gamma _3 x_0^T x_0,
    \end{align}
    and therefore,
    \begin{align}
        \mathbb{E}[x(t)^T  x(t) ] \le \max \{\frac{\gamma _3}{\gamma _2} x_0^T x_0, \frac{\gamma_3 c_1}{\gamma_1 \gamma_2} \},~ \forall t
    \end{align}
\end{remark}

\subsection{Stochastic Stabilizing Control}
Proposition~\ref{thm1} gives a sufficient condition for the system~\eqref{eq_1} to be stochastically stabilizable with the given controller~\eqref{eq_5}. However, it would be of practical interest to not only analyze stability for a given control system but design a stochastic stabilizing controller (SSC). Such a SSC design issue shall be addressed in this subsection. 

With the definition of $A_{cl}(i)$, the stability condition \eqref{eq_8} can be written as follows: for each $i \in \mathcal{I}$
\begin{align}
     [A(i)+B(i)&K(i)C(i)]^T P(i) + P(i)  [A(i)+B(i)K(i)C(i)] \nonumber  \\
     &+  \sum _{j=1}^{N} q _{ij} P(j)< 0  \label{eq_17}
\end{align}

The goal of the SSC design now can be stated as: find $P(i) >0$ and $K(i)$ such that the matrix inequality \eqref{eq_17} holds. Unfortunately, it is not really easy to solve the feasibility problem of \eqref{eq_17} due to the coupling of $K(i)$ and $P(i)$. To overcome this difficulty, a convex relaxation is introduced in below.
    
Let $S(i) = P(i)^{-1}$ and pre- and post-multiply by $S(i)$, leading to an equivalent condition as 
\begin{align}
    & S(i) [A(i)+B(i)K(i)C(i)]^T + [A(i)+B(i)K(i)C(i)] S(i) \nonumber \\
    & \qquad \qquad + S(i) \left[ \sum _{j=1}^{N} q _{ij} S(j)^{-1} \right] S(i) < 0, \label{eq_18}
\end{align}
and therefore,
\begin{align}
    & S(i)A(i)^T + S(i) C(i)^TK(i)^T B(i)^T + B(i)K(i)C(i)S(i)   \nonumber \\
    & \qquad + A(i)S(i)  + S(i) \left[ \sum _{j=1}^{N} q _{ij} S(j)^{-1} \right] S(i) < 0 \label{eq_19}
\end{align}

Define $\Lambda_i(S)$ and $\Xi_i(S) $ as follows
\begin{align}
    \Lambda_i(S) &= [\sqrt{q_{i1}} S(i), \cdots, \sqrt{q_{ii-1}} S(i), \sqrt{q_{ii+1}} S(i)  
 \\
 & \quad \quad \cdots, \sqrt{q_{iN}} S(i)] \nonumber \\
 \Xi_i(S) &= \mathrm{diag}\left(S(1), \cdots, S(i-1), S(i+1), \cdots, S(N)\right)
\end{align}
Then the term $S(i) \left[ \sum _{j=1}^{N} q _{ij} S(j)^{-1} \right] S(i)$ can be rewritten as 
\begin{align}
    S(i) \left[ \sum _{j=1}^{N} q _{ij} S(j)^{-1} \right] S(i) = q_{ii} S(i) + \Lambda_i(S)\Xi_i^{-1}(S)  \Lambda_i(S)^T
\end{align}

Letting $F(i) = K(i)Y(i) $ with $C(i)S(i)=Y(i)C(i)$ and using $\Lambda_i(S)$ and $\Xi_i(S)$, \eqref{eq_19} can be rewritten as
\begin{align}
     &S(i)A(i)^T +  C(i)^T F(i)^T B(i)^T + A(i)S(i) + B(i)F(i)C(i) \nonumber \\
     & \quad \quad q_{ii} S(i) + \Lambda_i(S)\Xi_i^{-1}(S)  \Lambda_i(S)^T < 0 \label{eq_23}
\end{align}
Applying Schur complement, \eqref{eq_23} is equivalent to the following linear matrix inequalities (LMIs):
\begin{align}
    &\begin{bmatrix}
\Delta(i)  & \Lambda_i(S)\\ 
 \Lambda_i(S)^T & -\Xi _i(S)
\end{bmatrix} <0, \nonumber \\
& C(i)S(i)-Y(i)C(i) = 0 \label{eq_24}
\end{align}
where $\Delta(i) = S(i)A(i)^T +  C(i)^T F(i)^T B(i)^T + A(i)S(i) + B(i)F(i)C(i)+ q_{ii} S(i) $.

The above SSC design procedure can be summarized as the following proposition.
\begin{proposition} \label{thm2}
    A Markovian jump system as described in~\eqref{eq_1} can be made UBS in mean square at the origin with the control law~\eqref{eq_5} if there exist matrices $S(i) > 0$, $Y(i) >0$, and $F(i)$ ($i \in \mathcal{I}$) such that LMIs~\eqref{eq_24} hold. In particular, the corresponding SSC gain to be determined can be recovered as $K(i) = F(i)Y^{-1}(i)$.
\end{proposition}
\begin{proof}
    The proof is now straightforward, since the feasibility of \eqref{eq_24} implies the feasibility of \eqref{eq_17}. Therefore, by Proposition~\ref{thm1}, it yields the desired result. This completes the proof.
\end{proof}
\begin{remark}
    {In Proposition~\ref{thm2}, the SSC design problem is formulated as a feasibility check of a system of linear matrix inequalities \eqref{eq_17}. These LMI-based conditions are favored in modern control design, as they recast the control design to a semidefinite programming problem that can be efficiently solved using existing numerical algorithms. Besides, this approach provides a systematic way to design control for systems of arbitrary dimension. Additionally, the SSC design is carried out offline, thus  computationally efficient in practice.}
\end{remark}

\subsection{Performance Guaranteed Control}
The previous section provides an efficient way to synthesize SSC controllers for the system \eqref{eq_1} in terms of the feasibility issue of a set of LMIs, as presented in Proposition~\ref{thm2}. It is often quite relevant in reality to not only design controllers with a stability guarantee but also with some prescribed performance, such as a satisfied convergence rate and steady state accuracy. In the sequel, we will develop a novel approach to achieve performance guaranteed control (PGC) for the system~\eqref{eq_1}.

Recalling \eqref{eq_16}, it is noticed that the bound of the mean square error of the closed-loop relies on parameters $c_1$,  $\gamma _1$, $ \gamma_2$, and $ \gamma_3$, which can be viewed as a measure of the steady state performance for the system against the uncertain measurements. Unfortunately, it is usually intractable to directly optimize such a specification because of the non-convex nature of this optimization. To make it easier to solve, restrictions are imposed on some of the parameters, which gives rise to an upper bound for the original problem. While some conservativeness may be incurred, an efficient PGC synthesis method can be attained.

The constraints are put as follows for each $i \in \mathcal{I}$
\begin{align}
    & Q(i)   \le -\bar{\gamma}_1 P(i), \label{eq_26}\\
    & \lambda_{min}(P(i)) \ge \bar{\gamma}_2, \label{eq_27} \\
    & \lambda_{max}(P(i)) \le \bar{\gamma}_3, \label{eq_28} \\
    & \bar \gamma _3 = \alpha_1 \bar \gamma _2  \label{eq28}
\end{align}
where $\bar{\gamma}_1$, $\bar{\gamma}_2$,  $ \bar{\gamma}_3 >0$, and $\alpha_1 \ge 1$ are design parameters. As a result, an upper bound for \eqref{eq_16} can be obtained as 
\begin{align}
    \mathbb{E}[x(t)^T x(t)] & \le \frac{\gamma_3 c_1}{\gamma_1 \gamma_2}\le \frac{\alpha_1 c _1 }{\bar{\gamma}_1} \label{eq_31}
\end{align}

By the definition of $Q(i)$, \eqref{eq_26} can be written as
\begin{align}
     [A(i)+B(i)&K(i)C(i)]^T P(i) + P(i)  [A(i)+B(i)K(i)C(i)] \nonumber  \\
     &+  \sum _{j=1}^{N} q _{ij} P(j) \le -\bar{\gamma}_1 P(i)  \label{eq29}
\end{align}
It is easy to verify that the constraint \eqref{eq29} imposed is sufficient to imply the stability condition \eqref{eq_17}. Following the same math as presented in \eqref{eq_17}--\eqref{eq_23}, the ultimate LMIs are changed slightly to
\begin{align}
    &\begin{bmatrix}
\bar{\Delta}(i)  & \Lambda_i(S)\\ 
 \Lambda_i^{T}(S) & -\Xi _i(S)
\end{bmatrix} <0,\nonumber \\
& C(i)S(i)-Y(i)C(i) = 0 \label{eq30}
\end{align}
where $\bar{\Delta}(i) = S(i)A(i)^T +  C(i)^T F(i)^T B(i)^T + A(i)S(i) + B(i)F(i)C(i)+ q_{ii} S(i) + \bar\gamma_1 S(i)$.

Note also that in our LMI formulation, as seen in Proposition~\ref{thm2}, the decision variables are with respect to $S(i)$, $Y(i)$, and $F(i)$. Considering $S(i) = P(i)^{-1}$,  the constraints \eqref{eq_27} and \eqref{eq_28} therefore can be transformed into
\begin{align}
    & \lambda _{max} (S(i))  \le  \frac{1}{\bar{\gamma}_2}, \label{eq_29} \\
    &\lambda _{min} (S(i)) \ge  \frac{1}{\bar{\gamma}_3} \label{eq_30}
\end{align}

Due to the additional constraints enforced, the original optimization problem (i.e., minimizing \eqref{eq_16}) is now being reduced to minimize $c_1$ solely, and clearly, the resulting minimum is an upper bound for the original problem. The modified problem becomes
\begin{align}
    &\underset{K,S,Y,F}{\mathrm{minimize}} \quad c_1 = \max_{i \in \mathcal{I}} \{ \chi (i) \}\nonumber \\
    &\mathrm{subject ~ to} \quad \{ \eqref{eq30}-\eqref{eq_30}\} \nonumber\\
    & \qquad \qquad \quad S(i) > 0,~ Y(i) >0,  ~ i \in \mathcal{I} 
     \label{eq_34}
\end{align}

It is worth noting  that $K(i)$ and $P(i)$ are involved in the objective $c_1$, yet they are not the decision variables in our LMI formulation. Thus, some rearranging is required. By definition, we see that
\begin{align}
    \chi (i) &=  \mathrm{tr}( W(i) ^T P(i) W(i) ) \nonumber \\
    & = \mathrm{tr} \{ [B(i)K(i)D(i)]^T P(i)B(i)K(i)D(i) \} \label{eq_36}
\end{align}
Considering, together, the relation $C(i) S(i) -Y(i)C(i) = 0$, \eqref{eq_36} can be rewritten as
\begin{align}
    \chi (i) &=   \mathrm{tr} \{ [B(i) F(i)Y(i)^{-1} D(i)]^T P(i)  \nonumber \\
    & \qquad \quad B(i) F(i)Y(i)^{-1} D(i) \} \nonumber \\
    &=  \mathrm{tr}  \{ [B(i) F(i)C(i)S(i)^{-1}C(i)^{+} D(i)]^T P(i) \nonumber \\
    & \qquad \quad B(i) F(i)C(i)S(i)^{-1}C(i)^{+} D(i) \} \nonumber \\
    &=  \mathrm{tr}  \{ [B(i) F(i)P(i) D(i)]^T P(i) \nonumber \\
    & \qquad \quad B(i) F(i)P(i) D(i) \} \nonumber \\
    & \le  \bar{\gamma}_3 ^3 \mathrm{tr}  \{ [B(i) D(i)]^T B(i) D(i) F(i) ^T F(i) \} \label{eq_37}
\end{align}

Based on the upper bound \eqref{eq_37}, the performance guaranteed control synthesis can be formulated as the following convex constrained optimization:
\begin{align}
    &\underset{S,Y,F}{\mathrm{minimize}}  \quad \max_{i \in \mathcal{I}}  \mathrm{tr}  \{ [B(i) D(i)]^T B(i) D(i) F(i) ^T F(i) \} \nonumber \\
    &\mathrm{subject ~ to} \quad \{ \eqref{eq30}-\eqref{eq_30}\} \nonumber\\
    & \qquad \qquad \quad S(i) > 0, ~Y(i) >0, ~ i \in \mathcal{I} 
     \label{eq_38}
\end{align}
To make it more tractable, the above problem is equivalent to
\begin{align}
    &\underset{S,Y,F}{\mathrm{minimize}}  \quad {\gamma}_4 \nonumber \\
    &\mathrm{subject ~ to} \quad \{  \eqref{eq30}-\eqref{eq_30}\} \nonumber\\
    & \qquad \qquad \quad S(i) > 0, ~Y(i) >0,~ i \in \mathcal{I} \nonumber \\
    & \qquad \qquad \quad  \mathrm{tr}  \{ [B(i) D(i)]^T B(i) D(i) F(i) ^T F(i) \} \le {\gamma}_4,i \in \mathcal{I}
     \label{eq_39}
\end{align}
Denote ${\gamma}_4^*$ as the minimum for the problem \eqref{eq_39}. We state the following proposition to recap the PGC synthesis.
\begin{proposition} \label{thm:3}
     If the formulated PGC problem \eqref{eq_39} is feasible, then the system \eqref{eq_1} is UBS in mean square under the output feedback control law \eqref{eq_5} with the PGC gains $K^*(i) = F^*(i)Y^*(i)^{-1}$ that are recovered from the solution of the problem \eqref{eq_39}, and furthermore, the steady state accuracy can be guaranteed as
    \begin{align}
        \mathbb{E}[x(t)^T x(t)] & \le \frac{\alpha_1 \bar{\gamma}_3^3 {\gamma}_4^*}{\bar{\gamma}_1 }  \label{eq_41}
    \end{align}
\end{proposition}
\begin{proof}
    Since the problem \eqref{eq_39} is feasible, the LMI conditions~\eqref{eq30} holds, which implies the satisfaction of conditions~\eqref{eq_24}. Therefore, by Proposition \ref{thm2}, the resulting PGC controller is SSC and leads to the system \eqref{eq_1} to be bounded stable in mean square. In view of \eqref{eq_37} as well as the optimum of \eqref{eq_39}, $c_1 = \max_{i \in \mathcal{I}} \{ \chi (i) \} \le \bar{\gamma}_3^3 {\gamma}_4^*$. As a consequence, based on \eqref{eq28} and  \eqref{eq_31}, it yields the required result. This completes the proof.
\end{proof}

\begin{remark} \label{rmk:5}
    It is worth noting that the developed PGC design procedure not only ensures a satisfactory steady-state performance as shown in \eqref{eq_41}, but also allows to specify a guaranteed convergence rate through the parameter $ \bar{\gamma}_1$. This fact can be readily checked by Remark \ref{rmk:2}. 
\end{remark}
\begin{remark} \label{rmk:6}
    {Due to the conservative nature of PGC, the guaranteed bound relies on the selection of the design parameters $\bar{\gamma}_1$, $\bar{\gamma} _3$, and $\alpha _1$ (or $\bar \gamma _2$ due to constraint \eqref{eq28} ). These parameters are generally problem-specific and demand tuning in practice. A useful guideline is to first select a desired convergence rate for $\bar{\gamma _3}$, and then minimize the product $\alpha _1 \bar{\gamma}_1^3$, as it directly relates to the ultimate  steady-state bound in \eqref{eq_41}. It should be noted there is a fundamental trade-off between convergence rate and steady-state accuracy; in addition, faster convergence typically demands larger control efforts and scarifies comfort.}
\end{remark}

\section{adaptive Cruise control design} \label{sec4}
To demonstrate the obtained results, this section considers an adaptive cruise control (ACC) task for an AV.
The behavior of ego vehicle is modeled as a double integrator
\begin{align}
    \dot{x}^{ego}_1 &= x_2^{ego}, \nonumber \\ 
    \dot{x}^{ego}_2 &= u \label{eq1_}
\end{align}
where ${x}^{ego}_1\in \mathbb{R}$ and ${x}^{ego}_2\in \mathbb{R}$ are the position and the velocity of the ego vehicle, respectively, $u\in \mathbb{R}$ is the control input  to be designed. Accordingly, the leading vehicle is described by
\begin{align}
    \dot{x}^{ld}_1 &= x_2^{ld},  \nonumber \\ 
    \dot{x}^{ld}_2 &= 0 \label{eq2_}    
\end{align}
Note that without loss of generality the above constant velocity assumption can be  relaxed to the piecewise constant case.

Let $\tilde{x}_1 = {x}^{ego}_1 - {x}^{ld}_1 - \delta_1^d$ and $\tilde{x}_2 = {x}^{ego}_2 - {x}^{ld}_2 $; note that $\delta_1^{d}$ denotes the desired distance required to be maintained for safety purposes and $\delta_1 = {x}^{ego}_1 - {x}^{ld}_1 $ the actual relative distance between the leading and ego vehicles. 

Based on \eqref{eq1_} and \eqref{eq2_}, the error dynamics of the ACC system is obtained as
\begin{align}
    \dot{\tilde{x}}_1 &= \tilde{x}_2, \nonumber \\ 
    \dot{\tilde{x}}_2 &= u  \label{eq3_}
\end{align}
Denote $\tilde{x} = [\tilde{x}_1, \tilde{x}_2]^T$. \eqref{eq3_} can be written into
\begin{align}
    \dot{\tilde{x}} = A \tilde{x} + B u  \label{eq4_}
\end{align}
with $A$ and $B$ given as
\begin{align}
    A = \begin{bmatrix}
        0 & 1 \\
        0 & 0 
    \end{bmatrix}, \quad    
    B = \begin{bmatrix}
        0  \\
        1  
    \end{bmatrix}
\end{align}

We consider an unreliable measurement as represented by
\begin{align}
     y(t)&=  C(i)\tilde{x}(t) + D(i) \omega(t) \label{eq5_}
\end{align}
where $\omega \in \mathbb{R}^{2}$ stands for the measurement error that is independently Gaussian distributed with zero mean and unit variance. $C(i)$ and $D(i)$ are subject to random switching caused by the imperfect AI-based perception, which are given, respectively, as
\begin{align}
    C(i) &= \left\{\begin{matrix}
\mathrm{diag}(0,1),& i = 0 \\ 
 \mathrm{diag}(1,1), & i =1 \end{matrix} \right.   \label{eq6__}\\
    D(i) &= \left\{\begin{matrix}
\mathrm{diag}(\sigma _{00},\sigma _{01}),& i = 0 \\ 
 \mathrm{diag}(\sigma _{10},\sigma _{11}), & i =1 \end{matrix} \right. \label{eq6_}
\end{align}
Here, $i = r(t) \in \{0,1\}$ indicates the status of the perception system with $0$ for the misdetected case and $1$ for the normal case, and $r(t)$ is govern by a Markov chain with a given generator matrix $Q$ as
\begin{align}
    Q= \begin{bmatrix}
        q_{00} & q_{01} \\
        q_{10} & q_{11} 
    \end{bmatrix}
\end{align}
Note that $Q$ can be estimated off-line from experiments. Additionally, in \eqref{eq6__} the relative velocity estimate is still used even in the misdetected case, which is plausible as it is assumed in our setup that the leading vehicle is driving at constant speed, while the parameter $\sigma_{01}$ in $D(0)$ accounts for the uncertainty of this information.

It is easy to see that \eqref{eq4_} and \eqref{eq5_} together form the PEM-augmented dynamic model, and thus, the control law can be proposed as
\begin{align}
    u(t) = K(i) y(t) \label{eq7_}
\end{align}
Plugging in the control law \eqref{eq7_} into \eqref{eq4_} gives the following closed-loop system
\begin{align}
    d \tilde{x} (t) = [ A+ BKC(i) ]\tilde x(t) dt + B K D(i) dw(t) \label{eq8_}
\end{align}
and let $A_{cl}(i) =  A+ BKC(i)$, $W(i) = BKD(i)$. Then \eqref{eq8_} can be written as
\begin{align}
    d \tilde{x} (t) = A_{cl}(i) \tilde x(t) dt+ W(i) dw(t) \label{eq9_}
\end{align}

The stability of the above closed-loop system \eqref{eq9_} with given controllers can be checked by the sufficient condition presented in Proposition~1. Besides the control analysis, the SSC control design can be resolved by Proposition \ref{thm2},  and furthermore, a PGC cruise controller can be obtained by solving the problem \eqref{eq_39}. 

It is easy to verify when the misdetection occurs, the corresponding subsystem is actually uncontrollable; however, it is still possible to design controllers as presented to stochastically stabilize the overall ACC error system \eqref{eq3_} by making use of the properties of Markov chains.

% \section{observer-based control}
% In previous sections, a static control structure is used to stabilize the stochastic system \eqref{eq_1}, and a disturbance rejection synthesis is presented in order to ensure performance while disturbances applied. Such a strategy, however, may be slightly conservative due to the fact that the controller will take no actions when misdetection occurs, while the overall stochastic stability can be ensured. This could lead to the case that the ego vehicle may stop when the object is missing forever.   

% To overcome the above issue, this section studies an more active strategy, where an estimation process is involved in the control design. In this case an observer is developed to offer state estimates for the stochastic system, and control laws are then designed based on the result of the estimates.

% \subsection{Observer Design}
% We now consider the following system
% \begin{align}
%     \dot{x}(t) &= A x(t) + B u(t) \nonumber \\
%     y(t) &= \left\{\begin{matrix}
%  0,&  r(t) = 0 \\ 
%  x(t) + D \dot{w}(t), & r(t) =1 
% \end{matrix}\right. \label{eq_32}
% \end{align}
% where $y(t)$ represents the current output of the perception subsystem. Note that the measurement relies on the status of the the perception subsystem as represented by $r(t)$, which satisfies a continuous-time Markov process; here it is assumed when the perception is missed, the measurement is set to zero, and otherwise, the noisy measurement is used.

\section{experiments} \label{sec5}
To verify the effectiveness and robustness of the presented control design methods, this section carries out a series of comparative experiments under three distinct driving scenarios. The considered adaptive cruise control system is equipped with an AI-based perception system subject to both measurement errors and misdetection. The noise intensities are set as $\sigma_{00} = 1$, $\sigma_{01}=1$, $\sigma_{10} = 0.05$ and $\sigma_{11}=0.5$. The probability transition  rate matrix $Q$ is given as $q_{00} = -4$, $q_{01} = 4$, $q_{10} = 0.5$, and $q_{11} = -0.5$. A sample path of this Markov chain is shown in Fig.~\ref{fig:1}, and it can be seen that the misdetection phenomenon within the system occurs occasionally. The output feedback control law \eqref{eq7_} is used in experiments. By following Remark~\ref{rmk:6}, the design parameters are selected as $\bar{\gamma} _1 = 0.8$, $\bar{\gamma} _2 = 0.1$, and $\bar{\gamma} _3 = 1$, and solving problem \eqref{eq_39}, the PGC control gains are obtained as $K_{pgc}(0)= [0, -2.52] $ and $K_{pgc}(1)= [-2.61, -1.76]$. Three commonly used driving strategies, that is, intelligent driving model (IDM), rule based control (RBC), and active learning based control (LBC), are employed as comparison.  
\begin{figure}[!htbp]
\centering
\includegraphics[width=0.5\textwidth]{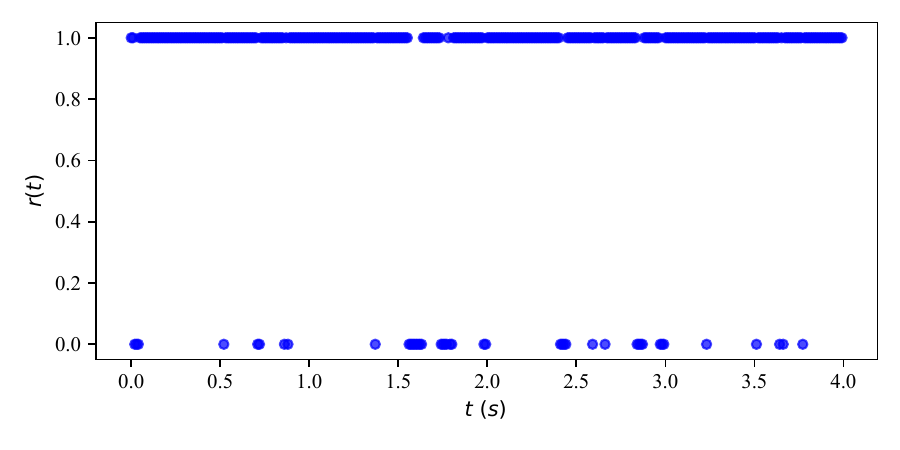}
\caption{Status of the AI-based perception system with a low misdetection rate.}
\label{fig:1}
\end{figure}
\begin{figure}[!htbp]
\centering
\includegraphics[width=0.5\textwidth,trim=0.2cm 0.2cm 0cm 0.2cm, clip]{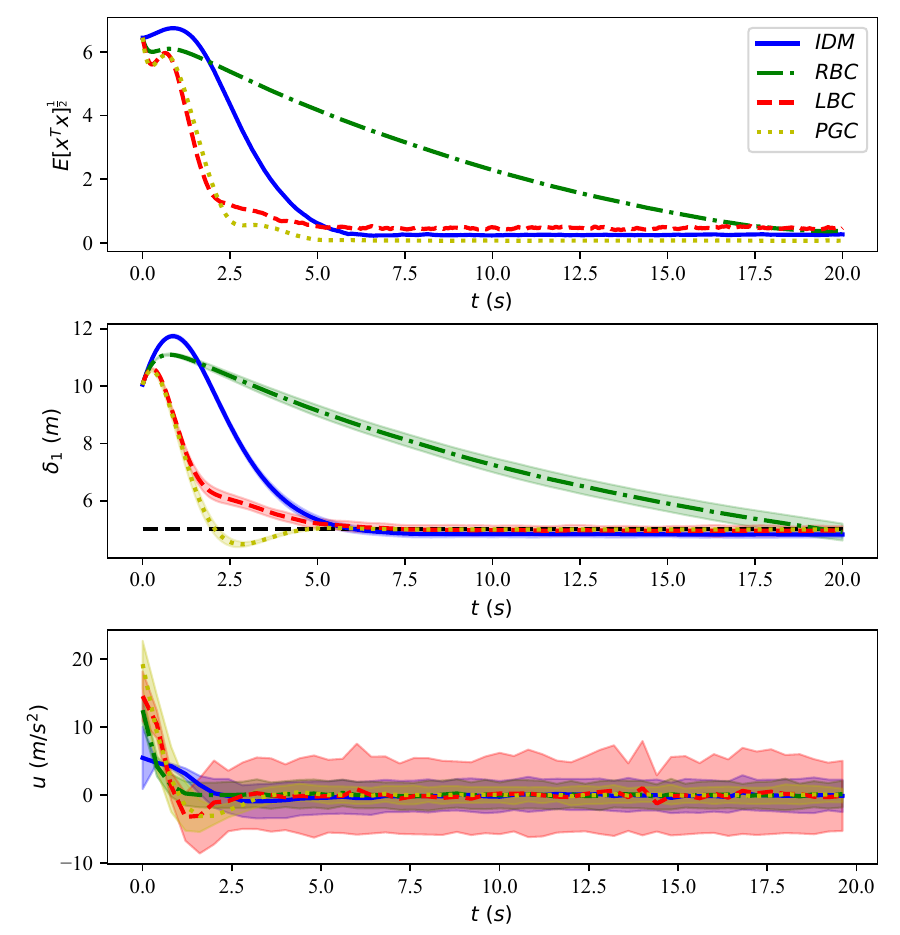}
\caption{ACC performance with a constant speed leading vehicle and a low misdetection rate.}
\label{fig:2}
\end{figure}

The desired relative distance between leading and ego vehicles are set as $\delta _1^d =- 5$. The initial states of the ACC system are given as $[x_1^{ego}(0),x_2^{ego}(0)]^T = [0,1]^T$, $[x_1^{ld}(0),x_2^{ld}(0)]^T = [10,5]^T$. In the first driving scenario, we assume the leading vehicle travels at a constant speed, representing a relatively good driving condition. The results are presented in Fig.~\ref{fig:2}, where the shaded regions denote the standard deviation over 500 sample runs for each method. It can be observed that all four driving strategies in this nominal case are able to stabilize the ACC system, yet the performance differs considerably. As shown in the first two subplots, the proposed PGC and the LBC approaches exhibit faster convergence compared to the IDM and, in particular, the rule based control. In addition, PGC maintains the best steady-state performance among the four approaches, as clearly seen in the first subplot.   Regarding driving comfort, as shown in the third subplot, the accelerating signal $u$  produced by the LBC method is noticeably noisy, leading to the most jerky and inefficient ride. In contrast, the PGC approach yields the most comfort and smooth driving behavior, while IDM and RBC fall in between. Therefore, it is concluded that PGC behaves the best in terms of convergence, steady-state accuracy, and driving comfort.

% \begin{figure}[!htbp]
% \centering
% \includegraphics[width=0.5\textwidth,trim=0.85cm 0.3cm 1.2cm 1.4cm, clip]{control_actions_shaded.pdf}
% \caption{Control actions.}
% \label{fig:3}
% \end{figure}
\begin{figure}[!htbp]
\centering
\includegraphics[width=0.5\textwidth,trim=0.2cm 0.1cm 0cm 0.2cm, clip]{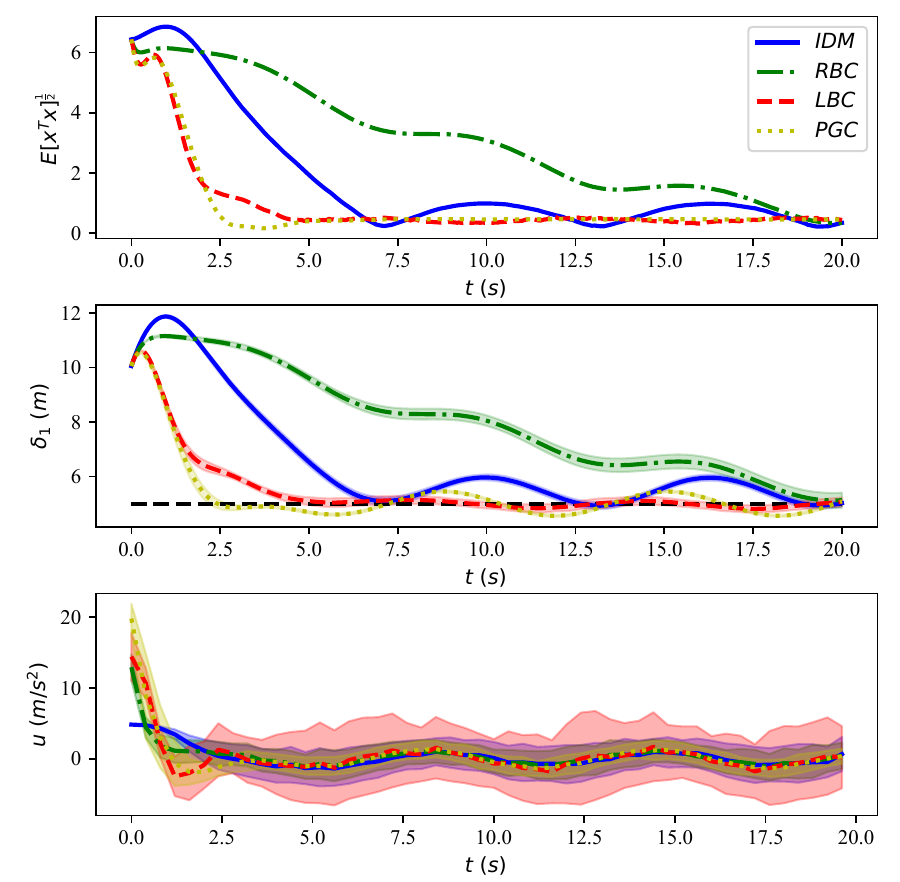}
\caption{ACC performance with a non-constant speed leading vehicle and a low misdetection rate.}
\label{fig:3}
\end{figure}
In the second scenario, we introduce uncertainty in the leading vehicle by setting its acceleration to $\dot{x}_2^{ld} = \sin{t} $. This setup simulates a more congested traffic environment, where the behavior of the leading vehicle becomes less predictable. The results are shown in Fig.~\ref{fig:3}. Similar to the constant speed case, both the PGC and LBC approaches demonstrate faster convergence than the IDM and RBC. It is interesting to see that the learning-based approach achieves as good steady-state performance as PGC, whereas significant performance degradation is observed in both IDM and RBC approaches, as illustrated in the first two subplots. Notably, the PGC approach continues to exhibit the smallest variation in acceleration, while the LBC method as before results in significantly more erratic acceleration behavior.  Hence, it can be concluded that the PGC can maintain its safety and performance, even in the presence of environmental uncertainties.

% In addition to the comparison of performance over different controllers, it is also studied how the parameters chosen(i.e., $\bar \gamma _1$, $\bar \gamma _2$, and $\bar \gamma _3$) in the PGC method affect the performance. The results are shown in Figs. \ref{fig:5}--\ref{fig:7}. As the increase of the $\bar \gamma _1$, it is observed from Fig. \ref{fig:5} that the convergence rate first grows and then decreases, while not varying that much. Such a result may be largely captured by the Remark \ref{rmk:2}. In addition, Fig. \ref{fig:6} shows that the performance could become remarkably worse with the growing of $\bar\gamma _2$. This is because employing larger values of $\bar\gamma _2$ leads to a more stringent feasible set as shown in \eqref{eq_29} or equivalently \eqref{eq_32}, thus resulting in a worse performance. On the other hand, it appears that the selection for $\gamma _3$ is relatively insensitive to the performance, as illustrated in Fig. \ref{fig:7}, which means the constraint on $\bar\gamma_3 $ is inactive when greater than 1. 
% \begin{table}[htbp]
%     \centering
%     \caption{}
%     \begin{tabular}{cccc}
%         \toprule
%         $\bar{\gamma}_1$ & $K(0)$ & $K(1)$& $\bar{\gamma}_4^*$ \\
%         \midrule
%         0.3 & [0, -0.56] & [-1.16,-0.75] & 0.56 \\
%         0.5 & [0, -1.31] & [-1.45,-1.18] & 1.31 \\
%         0.8 & [0, -2.52] & [-2.61,-1.76] & 4.02 \\
%         \bottomrule
%     \end{tabular}
%     \label{tab:1}
% \end{table}

\begin{figure}[!htbp]
\centering
\includegraphics[width=0.5\textwidth]{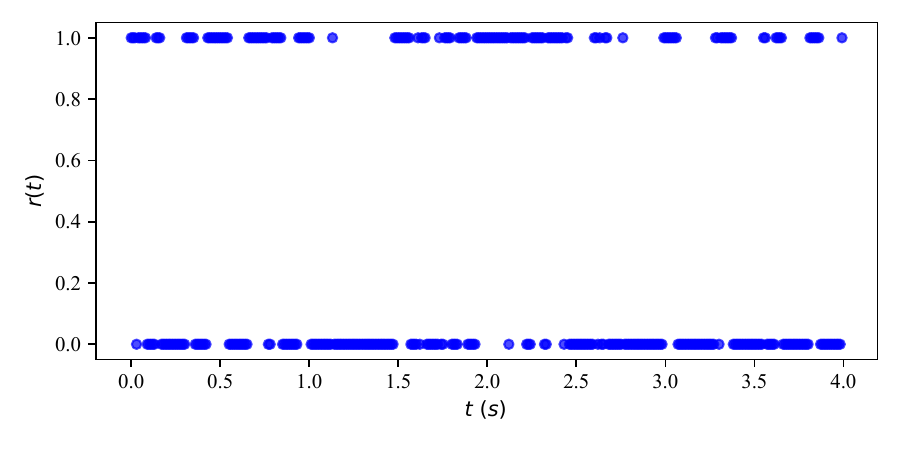}
\caption{Status of the AI-based perception system with a high misdetection rate.}
\label{fig:4}
\end{figure}

\begin{figure}[!htbp]
\centering
\includegraphics[width=0.5\textwidth,trim=0.2cm 0.1cm 0cm 0.2cm, clip]{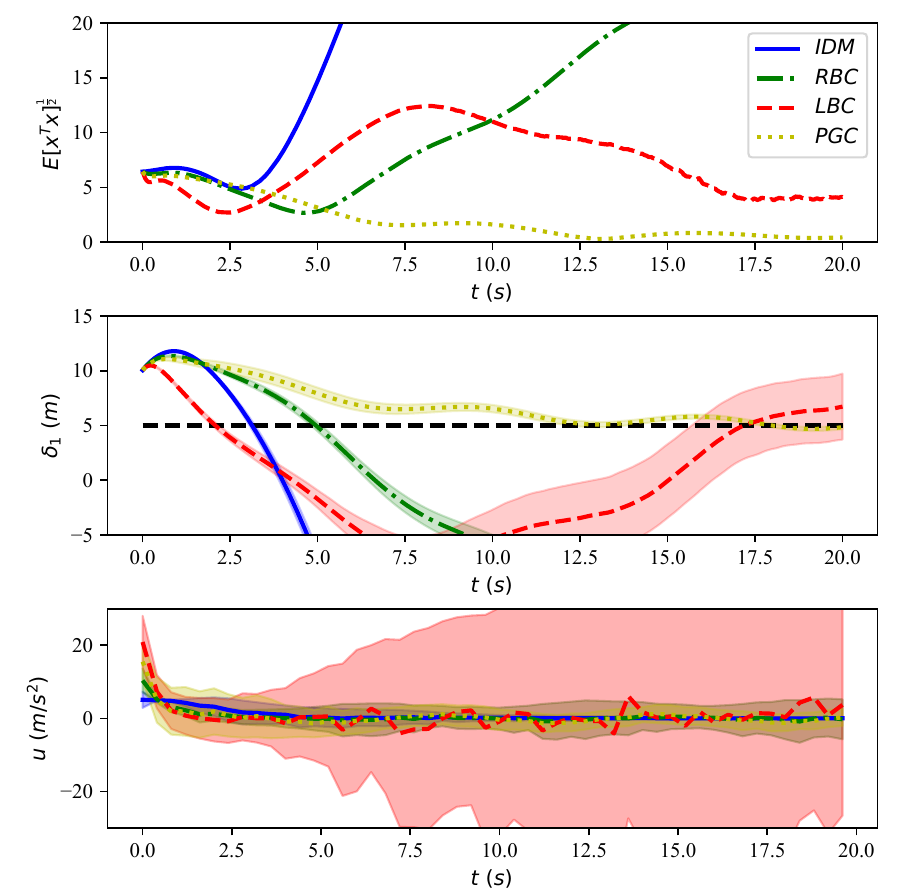}
\caption{ACC performance with a non-constant speed leading vehicle and a high misdetection rate.}
\label{fig:5}
\end{figure}
To further demonstrate robustness under extreme driving conditions (e.g., poor weather and low lighting), we simulate a scenario that includes not only an uncertain leading vehicle but also a perception system with a high misdetection rate. To do so, the transition rate matrix  $Q$ is set to $q_{00}=-4$, $q_{01}=4$, $q_{10}=3$, and $q_{11}=-3$. A typical sample path of this Markov chain is illustrated in Fig.~\ref{fig:4}, which indicates clearly a higher frequency of misdetected observations. The performance of four methods under this challenging scenario is presented in Fig.~\ref{fig:5}. The results show that only the proposed PGC method successfully stabilizes the ACC system and maintains the consistent behavior as the previous two cases, despite a minor degradation in convergence rate. However, other three methods all fail to stabilize the ACC, and particularly, collision occurs, i.e., $\delta_1 < 0$, around 4 and 7 seconds, respectively, as illustrated in the second subplot. Under a higher misdetection rate, IDM and rule based approaches are both diverging, while learning-based approach seems to have possibility to get back to the equilibrium but relies on an extremely unrealistic accelerating operation. Therefore, based on these discussions it is concluded that the proposed PGC approach demonstrates consistent and robust performance across three distinct scenarios. In particular, even under challenging conditions with both environmental uncertainty and a high misdetection rate, the PGC approach is able to maintain stability, safety, efficiency and comfort.
% \begin{figure}[!htbp]
% \centering
% \includegraphics[width=0.5\textwidth,trim=0.85cm 0.3cm 0.9cm 1.4cm, clip]{control_actions_acc_shaded.pdf}
% \caption{Control actions.}
% \label{fig:5}
% \end{figure}
% Overall, it is much clear that the tuning for control gains plays a significant role in attaining a satisfactory cruise control performance, especially when subject to the unreliable and uncertain perception system. The PGC method provides an effective way, through solving a convex optimization problem, to design controllers that achieve a moderate performance in terms of stabilizing control and noise suppression. Furthermore, we also study the relation between performance and parameters for the PGC method. Each parameter could have a definite interpretation on how it affects the performance or not, which is particularly important for the robustness of the design procedure.

% \begin{figure}[!htbp]
% \centering
% \includegraphics[width=0.5\textwidth,trim=0.5cm 0.3cm 0.9cm 1.4cm, clip]{gamma4_vs_gamma1.pdf}
% \caption{$\gamma _4*$ vs. $\gamma _1$.}
% \label{fig:4}
% \end{figure}

% \begin{figure}[!htbp]
% \centering
% \includegraphics[width=0.5\textwidth,trim=0.5cm 0.5cm 0cm 0.3cm, clip]{performance_vs_r1.pdf}
% \caption{ACC performance vs. $\gamma _1$.}
% \label{fig:5}
% \end{figure}

\section{conclusion} \label{sec6}
This paper offers a formal framework for modeling, analysis and synthesis  to address AI-based autonomous vehicles.  Based on recent developments of PEMs, two major classes of AI-induced perception errors are considered, that is, the misdetection that captures the intermittent failures in AI-driven sensing and perception as well as the measurement noise.There uncertainties are modeled using continuous-time Markov chains and Gaussian processes, respectively. Within this framework, analytical results are derived to characterize the closed-loop behavior of AI-based AVs. Furthermore, two output feedback control synthesis approaches are developed using LMIs.  The developed methods are applied to an ACC case study. Extensive experiments under various driving scenarios demonstrates that the proposed PGC strategy outperforms the traditional autonomous driving approaches in terms of both robustness and performance, when AI-induced uncertainties arise.

Despite these remarkable results, the presented framework has certain limitations. For example, the probability transition rate of the Markov chain as well as the covariance of Gaussian noise are known \textit{a priori}. This can be restrictive, as those information may vary depending on the weather, traffic, and lighting conditions.  In addition, other perception error patterns can appear in different AI-based systems, for example, perception bias. These limitations highlight practical challenges and motivate valuable future work. 

\appendix

\section{Appendix A}
In continuous-time Markov chain, the switching between two states is characterized by the following conditional probability transitions
\begin{align}
    {P}(r(t+h) =j | r(t) =i) = \begin{cases}
        q_{ij} h + o(h), &  j \ne i \\
        1 + q_{ii} h + o(h), & \text{otherwise} \label{apdx0}
    \end{cases}
\end{align}
with $q_{ij} \ge 0$ when $j \ne i$, $q_{ii} = -\sum_{j\ne i} q_{ij}$, and $\lim _{h\to 0} o(h)/h =0$.

\begin{definition} \label{def3}
    The infinitesimal operator $\mathcal{L}$ of $V(x(t) , r(t))$,  also called the expected change rate of $V(\cdot)$ at point $(t, x(t) ,r(t) = i )$, is defined by
    \begin{align}
        & \mathcal{L} V(x(t) , r(t)) \nonumber \\
        &  = \lim _{h\to 0} \frac{1}{h} \left[ \mathbb{E}[V(x(t+h),r(t+h) - V(x,i) | x(t)=x, r(t)=i) ]\right] \label{apdx1}
    \end{align}
\end{definition}
For the above mathematical expectation, two possible events exist: first, for $\tau \in (t, t+h]$, $r(\tau ) = i$, i.e., the Markov chain remains at state $i$, and therefore, in this case, $V(x(\tau),r(\tau)) - V(x,i) = 0$. Second case, on the other hand, jumps from state $i$ to state $j$ ($j \ne i$). Based on this observation and the Markov transition rate, one can easily obtain
\begin{align}
    &\mathbb{E}[V(x(t+h),r(t+h) - V(x,i) | x(t)=x, r(t)=i) ] \nonumber \\
    &=\mathcal{L}_{x} V(x(t),i) h + \sum _{j\ne i} q_{ij} [V(x,j) -V(x,i)] h +  o(h) \nonumber \\
    &= \mathcal{L}_{x} V(x(t),i) h + \sum _{j\ne i} q_{ij} V(x,j) h- \sum _{j\ne i} q_{ij} V(x,i) h +  o(h)\nonumber \\
    &= \mathcal{L}_{x} V(x(t),i) h + \sum _{j\ne i} q_{ij} V(x,j) h+ q_{ii} V(x,i) h +  o(h)\nonumber \\
    &= \mathcal{L}_{x} V(x(t),i) h  + \sum _{j \in \mathcal{I}} q_{ij} V(x,j) h +  o(h)  \label{apdx2}
\end{align}
Substituting \eqref{apdx2} into \eqref{apdx1} yields
\begin{align}
    \mathcal{L} V(x(t) , r(t)) = \mathcal{L}_{x} V(x(t),i)  + \sum _{j \in \mathcal{I}} q_{ij} V(x,j) \label{apdx3}
\end{align}
Notice that $x(t)$ is also an Ito process, therefore the following well-known lemma is introduced in order to handle the term $\mathcal{L}_{x} V(x(t),i) $.
\newtheorem{lemma}{Lemma}
\begin{lemma}
    Let $f(x(t), t)$ be a sufficiently smooth function with an Ito process $x(t)$ defined by
    \begin{align*}
        dx(t) = \mu (x(t),t) dt + \sigma (x(t),t) dw(t)
    \end{align*}
    then the differential of $f(x(t),t)$ is given as
    \begin{align*}
        d f(x(t), t) &= \left( \frac{\partial f}{\partial t} + \mu (x,t) \frac{\partial f}{\partial x}  + \frac{1}{2} \sigma^2 (x,t) \frac{\partial ^2 f}{\partial x^2}  \right) dt \\
        & ~ + \sigma (x(t),t)\frac{\partial f}{\partial x} dw(t)
    \end{align*}
\end{lemma}

Applying the above lemma and the fact $\mathbb{E}[dw(t)] =0 $, the infinitesimal operator for $f(x(t),t)$ with the underlying Ito process $x(t)$ can be obtained as
\begin{align}
    \mathcal{L} f(x(t) , t) = \frac{\partial f}{\partial t} + \mu (x,t) \frac{\partial f}{\partial x}  + \frac{1}{2} \sigma^2 (x,t) \frac{\partial ^2 f}{\partial x^2} \label{apdx4}
\end{align}

Therefore, using the above result, \eqref{apdx3} can be further calculated as
\begin{align}
    \mathcal{L} V(x(t) , r(t)) &= \mu (x,t) \frac{\partial V}{\partial x}  + \frac{1}{2} \sigma^2 (x,t) \frac{\partial ^2 V}{\partial x^2} \nonumber\\ 
    & \quad+ \sum _{j \in \mathcal{I}} q_{ij} V(x,j) \label{apdx5}
\end{align}

\bibliographystyle{IEEEtran}
\bibliography{reference}

\end{document}